\documentclass[11pt]{article}

\usepackage[letterpaper,margin=1in]{geometry}
\usepackage[utf8]{inputenc} % allow utf-8 input
\usepackage[T1]{fontenc}    % use 8-bit T1 fonts
\usepackage[colorlinks=true]{hyperref}       % hyperlinks
\usepackage{url}            % simple URL typesetting
\usepackage{booktabs}       % professional-quality tables
\usepackage{amsfonts}       % blackboard math symbols
\usepackage{amsmath,dsfont,amsthm,mathtools,amssymb}

\newcommand{\calB}{\mathcal{B}}
\newcommand{\calN}{\mathcal{N}}
\newcommand{\calP}{\mathcal{P}}

\renewcommand{\O}{\mathsf{O}}
\newcommand{\snr}{\mathsf{snr}}
\newcommand{\mmse}{\mathrm{mmse}}

\newcommand{\bbE}{\mathbb{E}}
\newcommand{\bbR}{\mathbb{R}}

\renewcommand{\d}{\,\mathrm{d}}
\newcommand{\Hmat}{\mathsf{H}}
\newcommand{\binent}{h_2}
\newcommand{\indic}{\mathds{1}}

\newcommand*\markov{\mathrel{-\mkern-2mu{\circ}\mkern-2mu-}}

\newtheorem{theorem}{Theorem}

\newtheorem{lemma}{Lemma}
\newtheorem{proposition}{Proposition}
\newtheorem{fact}{Fact}

\DeclarePairedDelimiter\floor{\lfloor}{\rfloor}

\DeclareMathOperator*{\argmax}{arg\,max}

\DeclareMathOperator*{\supp}{supp}

\hypersetup{
	colorlinks,
	linkcolor=blue,
	citecolor=red,
	urlcolor=blue
}

\usepackage[
backend=biber,
style=alphabetic,
maxbibnames=15,
maxalphanames=5,
minalphanames=3,
doi=false,
isbn=false,
url=false,
eprint=false,
backref=true,
]{biblatex}
\title{On Entropy-Constrained Gaussian Channel Capacity via the Moment Problem} 

\author{%
	Adway Girish$^*$, Shlomo Shamai (Shitz)$^\dagger$, Emre Telatar$^*$ \\
	{\small    $^*$School of Computer and Communication Sciences, EPFL}\\
	{\small    $^\dagger$Faculty of Electrical and Computer Engineering, 
	Technion}\\
	{\small   
	\texttt{adway.girish@epfl.ch},\,\texttt{sshlomo@ee.technion.ac.il},\,\texttt{emre.telatar@epfl.ch}}
}

\begin{document}

	\maketitle
	
	\begin{abstract}
		We study the capacity of the power-constrained additive Gaussian 
		channel with an entropy constraint at the input. In particular, we 
		characterize this capacity in the low signal-to-noise ratio regime at 
		small entropy. This follows as a corollary of the following general 
		result on a moment matching problem: We show that for any continuous 
		random variable with finite moments, the largest number of initial 
		moments that can be matched by a discrete random variable of 
		sufficiently small but positive entropy is three.
	\end{abstract}

	\section{Introduction}
	
	Consider an additive Gaussian noise channel with an input-output 
	relationship given by $Y = \sqrt{\snr}X + Z$, where $\snr > 0$ is the 
	signal-to-noise ratio (SNR) of the channel and $Z \sim \calN(0,1)$ is a 
	standard Gaussian random variable, independent of $X$. We denote the mutual 
	information between $X$ and $Y$ by $I(X, \snr) \triangleq I(X;Y)$, since we 
	have $Y = \sqrt{\snr}X + Z$ throughout this paper. The capacity of this 
	channel with a power constraint on the input is given by
	\begin{equation*}
		C(\snr) = \sup_{\bbE[X^2] \leq 1} I(X, \snr) = \frac12 \log ( 1 + \snr),
	\end{equation*}
	where the supremum is over all distributions of $X$ over $\bbR$. The 
	supremum is achieved by taking $X \sim \calN(0,1)$. Operationally, the 
	capacity $C(\snr)$ is the largest communication rate at which an 
	arbitrarily small error probability can be ensured over an additive 
	Gaussian channel of SNR equal to $\snr$. 
	
	In modern cloud-based communication networks, it is a remote agent that 
	determines the channel input, and the communication from the agent to the 
	transmitter is rate-limited.  Consequently, the channel input is entropy 
	constrained.  As a mathematical model of such a scenario, we propose to 
	study the \emph{entropy-constrained capacity} of the Gaussian channel, 
	defined as 
	\begin{equation}
		C_H(h, \snr) = \sup_{\substack{\bbE[X^2] \leq 1,\\H(X) \leq h}} I(X, 
		\snr), \label{eqn: C_H}
	\end{equation}
	for $h > 0$. The operational interpretation of $C_H(h, \snr)$ is that it is 
	the largest communication rate at which an arbitrarily small error 
	probability can be ensured over the Gaussian channel of SNR $\snr$ using an 
	input of entropy at most $h$.
	Note that the entropy constraint forces $X$ to have a discrete 
	distribution. Also note that without the power constraint $\bbE[X^2] \leq 
	1$, $C_H(h, \snr)$ would be equal to $h$ for all $h > 0$. 
	
	Computing the Gaussian channel capacity can be viewed as an approximation 
	problem of random variables in the following sense: solving $\sup_X I(X, 
	\snr)$ is equivalent to solving $\inf_X C(\snr) - I(X, \snr) = \inf_X D(X+Z 
	\,\|\, X_G + Z)$, where $Z \sim \calN(0,1)$, $X_G$ is Gaussian with the 
	same mean and variance as $X$, and $Z$ is independent of $X$ and $X_G$. 
	Thus, the problem is to approximate the Gaussian random variable $X_G$ by 
	an $X$ that minimizes the ``non-Gaussianity'' $D(X+Z \,\|\, X_G + Z)$.
	Our main result characterizes $C_H(h, \snr)$ in the low-SNR regime, i.e., 
	$\snr \to 0$, for sufficiently small $h$. It turns out that in this regime, 
	the approximation problem described above essentially takes on the 
	following form, which can be posed independently of the above, in terms of 
	just the moments of the random variables involved.
	
	Consider the classical mathematical problem of approximating continuous 
	distributions with discrete 
	distributions \emph{by matching as many of their initial moments as 
	possible}, called the moment problem 
	\cite{akhiezer1965classical,Schmudgen2017book}. 
	In particular, we are interested in identifying discrete distributions that 
	match as many initial moments as possible \emph{while having a sufficiently 
	small 
		entropy}. Moments determine the 
	tail behaviour of distributions, but it is known that distributions that 
	are 
	visually very different can have nearly identical moment generating 
	functions, 
	and thus, moments \cite{mccullagh1994does}. Continuous distributions can be 
	thought of as having 
	``infinite'' entropy. Thus, a natural question to ask is whether we can 
	replicate similar tail behaviours as infinite entropy continuous 
	distributions 
	using ``low-entropy'' discrete distributions.
	We show that discrete distributions with sufficiently 
	small entropy (i.e., smaller than some positive constant depending on the 
	continuous distribution) can match no more than \emph{three} moments with 
	\emph{any} 
	continuous 
	distribution of finite moments. 
	
	Using this result, we show that for any $h < \binent(1/3) \approx 0.92$ 
	bits, as $\snr \to 0$, $C_H(h, \snr)$ is at most a constant factor times 
	$\snr^4$ away from $C(\snr)$.
	In addition to this low-SNR regime, we provide asymptotic expressions for 
	$C_H(h, \snr)$ in the regimes where $h \to 0$ and $h \to \infty$ in 
	Section~\ref{sec: results}.
	These follow as immediate corollaries of results on the asymptotic 
	tightness of $F_I$ curves \cite{calmon2018sdpi}. More details on the moment 
	problem and the solution of its low-entropy version are in 
	Section~\ref{sec: moments}. We conclude with some perspectives and open 
	problems in Section~\ref{sec: conc}.
	
	% \newpage
	
	\subsection{Prior work}
	
	% discrete distributions, practical constraints (cardinality, peak, moments)
	% entropy constraint
	% moment matching as distribution approximation.
	Several works have studied the effect of practical constraints on channel 
	capacity. For example,
	Wu and Verd\'u \cite{wu2010impact} considered a cardinality constraint on 
	the (discrete) input distribution and showed, among other things, that the 
	cardinality-constrained capacity approaches the unconstrained capacity 
	$C(\snr)$ exponentially fast as the cardinality goes to infinity.
	Surprisingly, a peak-amplitude constraint (instead of the power and entropy 
	constraints) also results in the capacity-achieving distribution being 
	discrete, as shown by Smith \cite{smith1971information}. This spurred 
	interest in characterizing the cardinality of the capacity-achieving 
	distribution and bounds on the capacity with an amplitude constraint 
	\cite{sharma2010transition,dytso2019capacity,thangaraj2017capacity}. 
	More recently, there have been extensions to various moment constraints and 
	channels 
	\cite{ma2021first,afaycal2001fading,lapidoth2009poisson,barletta2024binomial},
	 motivated by practical setups such as optimal communication.
	
	While our motivation for the entropy constraint comes from cloud-based 
	networks, it is worth noting that entropy constraints have become 
	increasingly popular in the machine learning community, particularly in 
	lossy source coding \cite{liu2022lossy,ebrahimi2024minimum}. 
	The moment problem has also been of interest in machine learning 
	\cite{li2015generative,nguyen2021distributional}, which is not unexpected, 
	given that it is an approximation problem. 
	Thus, identifying how many moments of continuous distributions can be 
	matched by low-entropy distributions is a timely and interesting problem in 
	its own right.
	The appearance of the moment problem in our context of computing the 
	Gaussian channel capacity is also not surprising, as it is known that 
	information measures associated with the Gaussian channel can be 
	approximated by polynomials of moments 
	\cite{guo2008estimation,alghamdi2024moments}.
	For an overview of the moment problem from a mathematical perspective, 
	refer to the recent textbook by Schm\"udgen \cite{Schmudgen2017book} (more 
	references are in Section~\ref{sec: moments}).

	\subsection{Notation}
	The sequence $(s_0, s_1, \dots)$ is represented by the shorthand 
	$(s_n)_{n=0}^{\infty}$. 
	Uppercase letters (e.g $X$, $Y$, \dots) denote random variables (random 
	variables). We use $\bbE[X]$ to denote the expectation of $X$.
	$D(X \,\|\, Y)$ denotes the KL divergence between the distributions of $X$ 
	and $Y$. The mutual information between $X$ and $Y$ is denoted by $I(X, 
	\snr)$, with $X$ and $Y$ being the input and output of an additive Gaussian 
	channel throughout the paper. The entropy of $X$ is denoted by $H(X)$, and 
	the binary entropy function by $\binent(x) = -x\log x - (1-x)\log (1-x)$. 
	All logarithms are taken with respect to any fixed base; when the base is 
	2, the unit is bits. We write $f(x) = \O(g(x))$ if there exists a finite 
	constant $M$ and $x_0 > 0$ such that $f(x) \leq M g(x)$ for all $|x| < 
	x_0$. We write $f(x) = \Theta(g(x))$ if $f(x) = \O(g(x))$ and $g(x) = 
	\O(f(x))$.
	The Hankel matrix of order $n$ is denoted by 
	$\Hmat_n(s_0,s_1,\dots,s_{2n})$ and is given by the $(n+1) \times (n+1)$ 
	matrix with $(i,j)$-th entry $s_{i+j}$ for $0 \leq i, j \leq n$
	(defined explicitly in \eqref{eqn: hankel}). $A \succ 0$ denotes that the 
	matrix $A$ is positive definite and $A \succeq 0$ that $A$ is positive 
	semidefinite. We use $\indic\{P\}$ to denote the indicator function of the 
	statement $P$, which is 1 when $P$ is true and 0 otherwise.

	\section{Asymptotic Characterizations of $C_H$} \label{sec: results}
	In this section, we describe our results characterizing $C_H(h, \snr)$ in 
	the following asymptotic regimes: 
	\begin{enumerate}
		\item[(i)] $h \to 0$ (Section~\ref{sec: entropy_asymp}, 
		Proposition~\ref{prop: F_I}),
		\item[(ii)] $h \to \infty$ (Section~\ref{sec: entropy_asymp}, 
		Proposition~\ref{prop: F_I}), and 
		\item[(iii)] $\snr \to 0$ and $h < \binent(1/3) \approx 0.92$ bits 
		(Section~\ref{sec: snr_asymp}, Theorem~\ref{thm: c_h}).
	\end{enumerate}
	Before doing so, we first show (Section~\ref{sec: existence}, 
	Proposition~\ref{prop: exist}) that the supremum in \eqref{eqn: C_H} is in 
	fact a maximum, i.e., for any $h, \snr > 0$, there is an input distribution 
	such that $X$ satisfies $\bbE[X^2] \leq 1$ and $H(X) \leq h$ and achieves 
	$I(X, \snr) = C_H(h, \snr)$. 
	
	\subsection{Existence of capacity-achieving distribution} \label{sec: 
	existence}
	We use arguments similar to Abou-Faycal et al.~\cite{afaycal2001fading} and 
	Wu and Verd\'u \cite{wu2010impact} to show the existence of a 
	capacity-achieving distribution. In particular, we show that the set of 
	feasible distributions is compact (with respect to the topology of weak 
	convergence; see, e.g., the book by Billingsley 
	\cite{billingsley2013convergence} for the results used in the proof) and 
	the mutual information restricted to this set is continuous, which implies 
	that the supremum is achieved \cite{luenberger1997optimization}.
	
	While we will not make use of this existence result in the remainder of 
	this paper, it is interesting to note that for any finite $h$, the capacity 
	is indeed achieved by some discrete distribution. This result also allows 
	us to assume without loss of generality that the discrete random variable 
	$X$ is zero mean and of unit variance, which may simplify further analysis.
	
	\begin{proposition} \label{prop: exist}
		For any $h, \snr > 0$, the entropy-constrained capacity defined in 
		\eqref{eqn: C_H} is given by
		\begin{equation*}
			C_H(h, \snr) = \max_{\substack{\bbE[X^2] \leq 1,\\H(X) \leq h}} 
			I(X, \snr),
		\end{equation*}
		and a maximizing choice of $X$ is zero mean, has unit variance, and has 
		entropy $h$.
	\end{proposition}
	\begin{proof}
		Let $\calP$ denote the Polish space of the set of all distributions on 
		$\bbR$ associated with the L\'evy metric, inducing the topology of weak 
		convergence.
		Let $\calP_h$ be the set of distributions on $\bbR$ satisfying the 
		constraints $\bbE[X^2] \leq 1$ and $H(X) \leq h$. Let 
		$(P_n)_{n=0}^{\infty}$ be a sequence of distributions in $\calP_h$ 
		converging weakly to $P \in \calP$. Let $X_n$ be distributed as $P_n$, 
		i.e., each $X_n$ satisfies $\bbE[X_n^2] \leq 1$ and $H(X_n) \leq h$ and 
		$X$ as $P$, i.e., $X_n$ converges to $X$ in distribution. Fatou's lemma 
		implies that $\bbE[X^2] \leq \liminf_{n \to \infty} \bbE[X_n^2] \leq 
		1$, and the lower-semicontinuity of entropy implies that $H(X) \leq 
		\liminf_{n \to \infty} H(X_n) \leq h$. 
		Hence, we have that $\calP_h$ is a closed subset of $\calP$. Further, 
		$\calP_h$ is tight, i.e., for every $\epsilon > 0$ there exists a 
		compact subset $K_{\epsilon}$ of $\bbR$ such that for all $P \in 
		\calP_h$, $X$ distributed as $P$ satisfies $\Pr\{X \notin 
		K_{\epsilon}\} \leq \epsilon$. Indeed, choosing $K_{\epsilon} = 
		\left[-1/{\sqrt{\epsilon}}, 1/{\sqrt{\epsilon}}\right]$, we have that 
		$\Pr\{X \notin K_{\epsilon}\} = \Pr\{|X| >1/{\sqrt{\epsilon}} \} \leq 
		\frac{\bbE[X^2]}{(1/\sqrt{\epsilon})^2} \leq \epsilon$, by Markov's 
		inequality. Thus, by Prokhorov's theorem, $\calP_h$ is weakly compact. 
		Since $\bbE[X^2] \leq 1$ for all distributions in $\calP_h$, we also 
		have that $I(X, \snr)$ is weakly continuous on $\calP_h$ 
		\cite{wu2012mmse}. Hence, the function achieves its maximum (which may 
		not be unique). 
		
		That the entropy of a maximizing $X$ is equal to $h$ follows by 
		observing that the optimization problem involves maximizing a concave 
		function over the complement of an open, convex set, hence any 
		maximizer has to lie on the boundary. We also have that the second 
		moment must be 1 for the same reason.
		To show that a maximizer must be of zero mean, suppose otherwise. 
		Observe that $\lambda(X - \bbE[X])$ with $\lambda = 1 / \sqrt{\bbE[X^2] 
		- \bbE[X]^2} > 1$ has a higher mutual information, mean zero and unit 
		variance.
	\end{proof}
	
	\subsection{Entropy asymptotics via $F_I$ curves} \label{sec: entropy_asymp}
	The results for the asymptotic regimes of $h \to 0$ and $h \to \infty$ 
	follow immediately from previously known results on $F_I$ curves, studied 
	by Calmon et al.~\cite{calmon2018sdpi}. These are defined as
	\begin{equation}
		F_I(h, \snr) = \sup_{\substack{\bbE[X^2] \leq 1,\\I(W;X) \leq h}} 
		I(W;Y), \label{eqn: F_I}
	\end{equation}
	with the supremum over all joint probability distributions over $W$ and 
	$X$, and $Y = \sqrt{\snr} X + Z$. The $F_I$ curves are a generalization of 
	the classical data processing inequality, which says that for $W \markov X 
	\markov Y$ forming a Markov chain, $I(W;X) \geq I(W;Y)$. This implies that 
	$F_I(h, \snr) \leq h$, but the $F_I$ curve gives us a finer 
	characterization of this decrease in mutual information. 
	
	There is a clear connection between $F_I$ and $C_H$, as setting $W = X$ in 
	\eqref{eqn: F_I} recovers \eqref{eqn: C_H}. This immediately implies that 
	$F_I$ is an upper bound to $C_H$, which is statement (i) of 
	Proposition~\ref{prop: F_I}.
	Note that we always have $C_H(h, \snr) \leq h$; statement (ii) shows that 
	this is tight as $h \to 0$, i.e., $\lim_{h \to 0} \frac{C_H(h, \snr)}{h} = 
	1$ and the difference $h - C_H(h, \snr)$ goes to $0$ as approximately 
	$h^{\frac \snr h}$. Similarly, we always have $C_H(h, \snr) \leq C(\snr) = 
	\frac12 \log(1 + \snr)$; statement (iii) shows that this is tight as $h \to 
	\infty$, i.e., $\lim_{h \to \infty} \frac{C_H(h, \snr)}{C(\snr)} = 1$ and 
	$C(\snr) - C_H(h, \snr)$ goes to $0$ doubly exponentially in $h$.
	\begin{proposition} \label{prop: F_I}
		The following statements are true for $C_H$ and $F_I$ (as defined in 
		\eqref{eqn: C_H} and \eqref{eqn: F_I} respectively), for all $\snr > 0$:
		\begin{enumerate}
			\item[(i)] $C_H(h, \snr) \leq F_I(h, \snr)$ for all $h > 0$.
			\item[(ii)] As $h \to 0$, $C_H(h, \snr) = h - e^{\snr\frac{\log 
			h}{h} + \O(\log\frac{\snr}{h})}.$
			\item[(iii)] As $h \to \infty$, $e^{-c_1(\snr) e^{4h}} \leq C(\snr) 
			- C_H(h, \snr) \leq c_2(\snr) e^{-c_3(\snr) e^{h}},$ for some 
			positive functions $c_1, c_2, c_3$.
		\end{enumerate}
	\end{proposition}
	\begin{proof} (i) follows immediately by observing that setting $W = X$ in 
	the definition of $F_I$ \eqref{eqn: F_I} yields exactly $C_H$. (ii) and 
	(iii) follow from the diagonal \cite[Remark 2]{calmon2018sdpi} and 
	horizontal \cite[Remark 5]{calmon2018sdpi} bounds on $F_I$ respectively, 
	which are attained by choosing $W = X$ in \eqref{eqn: F_I}.
	\end{proof}
	
	\subsection{Low-SNR asymptotics} \label{sec: snr_asymp}
	Our main result characterizing $C_H$ is in the low-SNR regime, as stated in 
	Theorem~\ref{thm: c_h} below. As $\snr \to 0$, we show that for any $h < 
	\binent(1/3) \approx 0.92$ bits, $C_H(h, \snr)$ is at most a constant times 
	$\snr^4$ worse than $C(\snr)$.
	\begin{theorem} \label{thm: c_h}
		As $\snr \to 0$, for any $h < \binent(1/3)$, we have $C(\snr) - C_H(h, 
		\snr) = \O(\snr^{4}).$
	\end{theorem} 
	\begin{proof}
		The main ingredient of the proof is Theorem~\ref{thm: main} 
		(Section~\ref{sec: moments}), which implies that for any $h < 
		\binent(1/3)$, among all discrete distributions with entropy at most 
		$h$, the largest $k$ such that the $k$-th moment of the discrete 
		distribution is equal to the $k$-th moment of the Gaussian 
		distribution, is three.
		
		Let us consider distributions that have a sufficiently large number of 
		finite moments (say $2n$, which is at least 10). By the I-MMSE 
		relationship \cite{guo2004immse}, note that $I(X, \snr)$ is $n$-times 
		differentiable in $\snr$ if $\bbE[X^{2n}]<\infty$ 
		\cite{guo2008estimation}. Hence, we can write a Taylor expansion of 
		$I(X, \snr)$ about $\snr = 0$ up to the $n$-th order ($n \geq 5$), with 
		the coefficients of $\snr^k$ being a polynomial of the first $k$ 
		moments of $X$ \cite{wu2010impact}. 
		Consider the difference $C(\snr) - I(X, \snr)$, which is equal to 
		$I(X_G, \snr) - I(X, \snr)$, as the capacity of the Gaussian channel 
		(with only a power constraint) is achieved by $X_G \sim \calN(0,1)$. 
		By Theorem~\ref{thm: main}, the first non-zero term in the Taylor 
		expansion of $C(\snr) - I(X, \snr)$ is $\Theta(\snr^4)$ for any $X$ 
		with finite $\bbE[X^n]$ up to $n \geq 10$. 
		Hence, the minimum of $C(\snr) - I(X, \snr)$ over all distributions 
		satisfying $\bbE[X^2] \leq 1$ and $H(X) \leq h$, which is exactly 
		$C(\snr) - C_H(h, \snr)$ is $\O(\snr^4)$. 
	\end{proof}
	
	Thus, the key to Theorem~\ref{thm: c_h} is the low-entropy moment problem, 
	which we state and solve in the next section.
	
	\section{The Moment Problem} \label{sec: moments}
	We now describe the classical moment problem in measure theory. Though the 
	problem is usually stated in terms of general measures, we restrict 
	ourselves to probability measures and continue to use the language of 
	random variables for consistency, since they are entirely equivalent: $\mu$ 
	is a probability measure on $\bbR$ (which we equip with the Borel 
	$\sigma$-algebra $\calB(\bbR)$ throughout) if and only if there is a random 
	variable $X$ such that $\Pr\{X \in A\} = \mu(A)$ for all $A \in 
	\calB(\bbR)$.

	One version of the classical \emph{moment problem} \cite{Schmudgen2017book} 
	is the following: 
	\begin{quote}
		Given a sequence $(s_n)_{n=0}^{\infty}$, does there exist a random 
		variable $X$ on $\bbR$ such that $\bbE[X^n] = s_n$ for all $n \geq 0$? 
	\end{quote}
	Hamburger \cite{hamburger1920erweiterung} characterized the sequences 
	$(s_n)_{n=0}^{\infty}$ with a positive answer to this question. The result 
	is stated in terms of the Hankel matrix of order $n$ associated with 
	$(s_n)_{n=0}^{\infty}$, given by
	\begin{equation}
		\Hmat_n(s_0, s_1,\dots, s_{2n}) = \begin{pmatrix}
			s_0 & s_1 & \dots & s_n \\
			s_1 & s_2 & \dots & s_{n+1} \\
			\vdots & \vdots & \ddots & \vdots \\
			s_n & s_{n+1} & \dots & s_{2n} \label{eqn: hankel}
		\end{pmatrix}.
	\end{equation}
	\begin{fact}[\cite{Schmudgen2017book,hamburger1920erweiterung}] Given an 
	infinite sequence $(s_n)_{n=0}^{\infty}$, there exists a random variable 
	$X$ on $\bbR$ such that $\bbE[X^n] = s_n$ for $n \geq 0$, if and only if 
	$s_0 = 1$ and the Hankel matrices $\Hmat_n(s_0,\dots, s_{2n})$ are positive 
	semidefinite for all $n \geq 1$.
		Further, $X$ has an infinite support if and only if $\Hmat_n(s_0,\dots, 
		s_{2n})$ is positive definite for all $n$.
		\label{fact: hamburger}
	\end{fact}
	
	Similar solutions can be obtained when the random variables are restricted 
	to be on various closed subsets of $\bbR$, such as the non-negative real 
	axis $[0, \infty)$ \cite{stieltjes1894recherches} and compact intervals 
	$[a,b]$ \cite{hausdorff1921summationsmethoden}. An interesting variant of 
	the problem is to require only a finite number of initial moments to be 
	matched, leading to the \emph{truncated moment problem}. Surprisingly, if 
	this can be done by any random variable, it can be done by a discrete 
	random variable with a finite number of atoms. The truncated sequences that 
	allow for a solution were characterized by Curto and Fialkow 
	\cite{CurtoFialkow1991}.

	\begin{fact}[\cite{Schmudgen2017book,CurtoFialkow1991}] Given a truncated 
	sequence $(s_n)_{n=0}^{k}$, there exists a random variable $X$ on $\bbR$ 
	such that $\bbE[X^n] = s_n$ for $n = 0,1,\dots,k$, if and only if $s_0 = 1$ 
	and
		\begin{enumerate}
			\item[(i)] (for odd $k = 2\ell+1$) there exists $\tilde 
			s_{2\ell+2}$ such that $\Hmat_{\ell+1}(s_0, s_1, \dots, s_k, \tilde 
			s_{2\ell+2}) \succeq 0$;
			\item[(ii)] (for even $k = 2\ell$) there exist $\tilde 
			s_{2\ell+1}$, $\tilde s_{2\ell+2}$ such that $\Hmat_{\ell+1}(s_0, 
			s_1, \dots, s_k, \tilde s_{2\ell+1}, \tilde s_{2\ell+2}) \succeq 0$.
		\end{enumerate} 
		Further, $X$ is discrete and has at most $\floor{k/2} + 1$ atoms.    
		\label{fact: hankel}
	\end{fact}
	
	\subsection{Low-entropy moment problem}
	Now suppose there is a target continuous random variable $W$ on $\bbR$ and 
	we wish to find a discrete random variable $X$ on $\bbR$ that approximates 
	$W$ by matching as many of their initial moments as possible. Let $s_n = 
	\bbE[W^n]$ for $n \geq 0$, then by Fact~\ref{fact: hamburger}, we have 
	$\Hmat_n(s_0, \dots, s_{2n}) \succ 0$ for all $n \geq 1$. By 
	Fact~\ref{fact: hankel}, this implies that there exists a discrete random 
	variable $X$ of at most $m+1$ atoms with moments $\bbE[X^n] = s_n$ for $n = 
	0,\dots,2m+1$. Thus, for any continuous random variable on $\bbR$, there 
	exists a discrete random variable on $\bbR$ with at most $m$ atoms that has 
	the same first $2m-1$ moments. By allowing for a large enough $m$, it is 
	possible to match an arbitrarily large number of moments using a discrete 
	(even finite) random variable. 
	
	In particular, consider the special case where the target continuous random 
	variable is Gaussian. It is known that the $m$-point Gauss--Hermite 
	quadrature distribution has the same first $2m-1$ moments as the Gaussian 
	distribution, and that no other discrete random variable with as many or 
	fewer atoms can do the same \cite[Theorem 2]{wu2010impact}, 
	\cite{gauss1814methodus,stoer1980introduction}. It is also known that for 
	large $m$, the entropy of the Gauss--Hermite quadrature distribution is 
	approximately $\frac{1}{2}\log m$, which grows unboundedly as $m \to 
	\infty$. This seems to suggest that we require a large entropy to match an 
	arbitrary number of moments, which begs the following general question:
	\begin{quote}
		Given a real-valued continuous random variable, how many moments can be 
		matched by a discrete random variable that has a ``sufficiently small'' 
		entropy? 
	\end{quote}
	We show that for any continuous random variable, at most three moments can 
	be matched by such a low-entropy discrete random variable. This is stated 
	formally below.
	
	\begin{theorem}
		For any continuous random variable $W$ with finite moments $m_n = 
		\bbE[W^n]$, there exists a positive number $\eta(W) < \frac12$ such 
		that for any $h \in (0, \binent(\eta(W))$, among all discrete random 
		variables $X$ such that $H(X) \leq h$, 
		the largest $k$ such that $\bbE[X^n] = m_n$ for $n = 1,2,\dots,k$ is 
		three. 
		In particular, when $W$ is symmetric, we have
		\begin{equation*}
			\eta(W) = \begin{cases}
				\frac{m_2^2}{m_4} & \text{if } m_4 \geq 3 m_2^2,\\
				\frac{5m_2^2 - m_4}{9m_2^2 - m_4} & \text{if } m_4 <3 m_2^2.
			\end{cases}
		\end{equation*}
		\label{thm: main}
	\end{theorem}
	
	For the special case of the Gaussian random variable $W \sim \calN(0,1)$, 
	$\eta(W) = \frac 13$ and 
	$\binent(\eta(W)) \approx 0.92$ bits.
	Before proceeding to the proof of Theorem~\ref{thm: main}, it is worth 
	clarifying that the theorem makes two claims, for any continuous random 
	variable $W$ with finite moments:
	\begin{enumerate}
		\item[(i)] If the discrete random variable $X$ is such that $\bbE[X^n] 
		= \bbE[W^n]$ for $n = 1,2,3,4$, then $H(X) \geq \binent(\eta(W))$.
		\item[(ii)] For any $h > 0$, there exists a discrete random variable 
		$X$ such that $H(X) \leq h$ and $\bbE[X^n] = \bbE[W^n]$ for $n
		= 1,2,3$, i.e., it is possible to match three moments of any continuous 
		random variable with a discrete random variable of arbitrarily low 
		entropy.
	\end{enumerate} 
	
	\subsection{Proof of Theorem 2} \label{sec: proof}
	
	We first characterize random variables with ``small'' entropy as being a 
	random combination of a probability mass at a single atom and some discrete 
	random variable with finite but not necessarily ``small'' entropy 
	(Lemma~\ref{lem: low-ent}). We then use this characterization, together 
	with Fact~\ref{fact: hankel}, to show that if the entropy is smaller than 
	some constant depending on the continuous random variable (e.g., 
	$\binent(1/3) \approx 0.92$ bits for the Gaussian distribution), we can 
	match at most the first three moments. Finally, we show that even if the 
	entropy is to be arbitrarily small, the first three moments can still be 
	matched.

	\begin{lemma} 
		For any $h \in (0,\log 2)$, the following statements are equivalent:
		\begin{enumerate}
			\item[(i)] The random variable $X$ satisfies $0 < H(X) \leq  h$.
			\item[(ii)] There exists $x_0 \in \bbR$, $\epsilon \in 
			\left(0,\frac{1}{2}\right)$ such that $\binent(\epsilon) \leq h$, 
			and a discrete random variable $\tilde X$ with $\Pr\{\tilde X = 
			x_0\} = 0$ (i.e.\ $x_0$ is not an atom of $\tilde X$) and entropy 
			$H(\tilde X) \leq \frac{h - \binent(\epsilon)}{\epsilon}$, such 
			that  
			% \begin{equation}
				%     X = \begin{cases}
					%         x_0 & \text{with probability } 1- \epsilon,\\
					%         \tilde X & \text{with probability }\epsilon.
					%     \end{cases} \label{eqn: lemma}
				% \end{equation}
			\begin{equation}
				X = U x_0 + (1-U) \tilde X, \label{eqn: lemma}
			\end{equation}
			where $U$ is a binary random variable independent of $\tilde X$, 
			taking values in $\{0,1\}$ with $\Pr\{U = 0\} = \epsilon$.
		\end{enumerate} \label{lem: low-ent}
	\end{lemma}
	\begin{proof}
		{(ii) $\implies$ (i): } Assume that there exist quantities $x_0, 
		\epsilon, \tilde X, U$ as given in statement (ii) and let $X$ be given 
		by \eqref{eqn: lemma}.
		For such $X$ and $\tilde X$, we can derive a relation between their 
		entropies as follows. Consider the joint entropy $H(X,U)$, which is 
		equal to $H(X) + H(U \mid X)$, by the chain rule of entropy. As 
		$\Pr\{\tilde X = x_0\} = 0$, we know that $U = 1$ if and only if $X = 
		x_0$, and hence $H(U \mid X) = 0$. The joint entropy is also equal to 
		$H(U) + H(X \mid U)$. The first term is equal to $\binent(\epsilon)$, 
		and the second term is equal to $\epsilon H(X \mid U = 0) + 
		(1-\epsilon) H(X \mid U = 1)$. When $U = 0$, we have $X = \tilde X$, 
		and hence $H(X \mid U = 0) = H(\tilde X)$. Similarly, when $U = 1$, we 
		have $X = x_0$, and $H(X \mid U = 1) = 0$. Putting everything together, 
		we have that $H(X) = H(X,U) = \binent(\epsilon) + \epsilon H(\tilde 
		X)$. Since $\binent(\epsilon) \leq h$ and $H(\tilde X) \leq \frac{h - 
		\binent(\epsilon)}{\epsilon}$, we have $H(X) \leq h$. Further, since 
		$\epsilon > 0$, $H(X) \geq \binent(\epsilon) > 0$, and we are done.
		
		{(i) $\implies$ (ii): } Let $X$ be a random variable with $H(X) \in (0, 
		h]$ for some $h \in (0, \log 2)$ and let $\supp(X)$ denote its support, 
		i.e., set of $x \in \bbR$ such that $\Pr\{X = x\} > 0$. Also let 
		$\epsilon = 1 - \max_{x \in \supp(X)} \Pr\{X = x\}$ and $x_0 = 
		\argmax_{x \in \supp(X)} \Pr\{X = x\}$. Note that this maximum is 
		well-defined even if the support is (countably) infinite, as the sum of 
		$\Pr\{X = x\}$ over all $x \in \supp(X)$ is 1. Clearly, $\epsilon \in 
		[0,1]$; we claim that $\epsilon$ must lie in $\left(0, 
		\frac{1}{2}\right)$. That $\epsilon > 0$ is trivial---if $\epsilon = 
		0$, we have $\Pr\{X = x_0\} = 1$, implying that $H(X) = 0$ which is a 
		contradiction. On the other hand, if $\epsilon \geq \frac12$, we have 
		that $\Pr\{X = x\} \leq \frac12$ for all $x \in \bbR$, and hence, $H(X) 
		= \sum_{x \in \supp(X)} \Pr\{X = x\} \log \frac{1}{\Pr\{X = x\}} \geq 
		\log 2 > h \geq H(X),$ which cannot be. Let $U = \indic\{X = x_0\}$ 
		such that $\Pr\{U = 0\} = \Pr\{X \neq x_0\} = \epsilon$. We define the 
		random variable $\tilde X$ as follows. If $X = x$ for some $x \neq 
		x_0$, set $\tilde X = x$. If $X = x_0$, randomly set $\tilde X$ to be 
		any value $x \in \supp(X) \setminus \{x_0\}$ with probability 
		$\frac{1}{\epsilon}\Pr\{X = x\}$. This choice makes $U$ and $\tilde X$ 
		to be independent, as $\Pr\{\tilde X = x \mid U = 0\}$ is equal to 
		$$\Pr\{\tilde X = x \mid X \neq x_0\} = \Pr\{X = x \mid X \neq x_0\} = 
		\frac{\Pr\{X = x\}}{\Pr\{X \neq x_0\}} = \frac{\Pr\{X = x\}}{\epsilon} 
		= \Pr\{\tilde X = x \mid X = x_0\},$$
		which is equal to $\Pr\{\tilde X = x \mid U = 1\}$. Note that 
		$\Pr\{\tilde X = x_0\} = 0$ and we always have $X = U x_0 + (1-U)\tilde 
		X$. By the same calculation as in the proof of (ii) $\implies$ (i), we 
		have that $H(X) = \binent(\epsilon) + \epsilon H(\tilde X) \leq h$, and 
		hence, $H(\tilde X)$ must be at most $\frac{h - 
		\binent(\epsilon)}{\epsilon}$. This completes the proof of the 
		equivalence of (i) and (ii).
	\end{proof}
	
	We are now ready to prove Theorem~\ref{thm: main}.
	We show that 
	\begin{enumerate}
		\item[(i)] any discrete random variable with entropy at most 
		$\binent(\eta(W))$ can match at most the first three moments with the 
		continuous random variable $W$, and 
		\item[(ii)] there exist random variables with arbitrarily small entropy 
		that can still match the first three moments. 
	\end{enumerate}
	Note that to simplify calculations, we may assume that $m_1 = 0$. Let $X' = 
	X - m_1$ and let $m'_n = \bbE[{X'}^n]$.  If $X$ has moments $m_n$, then 
	$m'_n = \sum_{i=0}^n (-1)^i {n \choose i} m_{n-i} m_1^{i}$, i.e., $m'_1 = 
	0$, $m'_2 = m_2 - m_1^2$, $m'_3 = m_3 - 3 m_1 m_2 + 2m_1^3$, and so on. 
	Thus, we assume without loss of generality that $m_1 = 0$, remembering that 
	we must replace $m_2$ by $m_2 - m_1^2$, $m_3$ by $m_3 - 3 m_1 m_2 + 
	2m_1^3$, and so on, in the final expression. However, since we only provide 
	an explicit expression for $\eta(W)$ in the case where $W$ is symmetric, we 
	have $m_1 = 0$ anyway. 
	
	\paragraph{Proof of (i): } The proof can be summarized as follows. We use 
	Lemma~\ref{lem: low-ent} to conclude that any random variable $X$ such that 
	$H(X) < \binent(\eta(W))$ is of the form \eqref{eqn: lemma} for some $x_0 
	\in \bbR$, $\epsilon < \eta(W)$ and discrete random variable $\tilde X$ 
	which has no mass at $x_0$. For $X$ to have moments $\bbE[X^n] = m_n$, 
	$\tilde X$ must have moments $\bbE[\tilde X^n] = s_n = \frac{m_n - 
	(1-\epsilon)x_0^n}{\epsilon}$. We then show that $\Hmat_2(1, s_1, \dots, 
	s_4)$ cannot be positive semidefinite for any choice of $\epsilon < 
	\eta(W)$. Thus, there is no choice of $\tilde s_5$ and $\tilde s_6$ that 
	makes $\Hmat_2(1, s_1, \dots, s_4, \tilde s_5, \tilde s_6)$ positive 
	semidefinite. By Fact~\ref{fact: hankel}, we have that no $\tilde X$ has 
	$s_1, \dots, s_4$ as the first four moments, and hence, no $X$ with $H(X) < 
	\binent(\eta(W))$ has $m_1, \dots, m_4$ as the first four moments.
	
	First note that if $H(X) = 0$, the only possibility is that $X = x$ with 
	probability 1 for some $x \in \bbR$. Since $\bbE[X] = x$, we must have $x = 
	m_1$ to match at least one moment. Since $W$ is a continuous random 
	variable, we have, in particular, that 
	$$\Hmat_1(1,m_1,m_2) = \begin{pmatrix}
		1 & m_1 \\
		m_1 & m_2
	\end{pmatrix}\succ 0,$$
	implying that $m_2 - m_1^2 > 0$. Hence, $m_2 \neq m_1^2 = \bbE[X^2]$, and 
	with $H(X) = 0$, we can match at most one moment. 
	
	Now suppose $0 < H(X) \leq h < \binent(\eta(W))$. 
	By Lemma~\ref{lem: low-ent}, $X$ must be of the form $U x_0 + (1-U) \tilde 
	X$ for some $\epsilon \in (0,\frac{1}{2})$ such that $\binent(\epsilon) 
	\leq h$, $x_0 \in \bbR$, $\tilde X$ with $\Pr\{\tilde X = x_0\} = 0$, and 
	$U \in \{0,1\}$ independent of $\tilde X$. Since $\binent(\epsilon) \leq h 
	< \binent(\eta(W))$ and $\eta(W) \leq \frac12$, we must also have $\epsilon 
	< \eta(W)$.
	Then, the $n$-th moment of $X$ can be written as
	\begin{align*}
		\bbE[X^n] &= \sum_{i=0}^n {n \choose i} \bbE\left[(U x_0)^i((1-U) 
		\tilde X)^{n-i}\right] = \sum_{i=0}^n {n \choose i} x_0^i\bbE\left[U^i 
		(1-U)^{n-i}\right] \bbE\big[ \tilde X^{n-i}\big],
	\end{align*}
	since $U$ and $\tilde X$ are independent. Note that $\bbE\left[U^i 
	(1-U)^{n-i}\right]$ is zero for all $i$ except $0$ and $n$.  Hence, the 
	above sum is simply
	\begin{align*}
		\bbE[X^n] &= \bbE[(1-U)^n]\bbE[\tilde X^n] + x_0^n\bbE[U^n] = \epsilon 
		\bbE[\tilde X^n] + (1-\epsilon) x_0^n.
	\end{align*}
	This implies that there exists $X$ such that $\bbE[X^n] = m_n$ for $n = 
	1,2,\dots,k$ if and only if there exists $\tilde X$ such that $\bbE[\tilde 
	X^n] = s_n$ for $n = 1,\dots,k$, with
	\begin{equation*}
		s_n = \frac{m_n}{\epsilon} -  \frac{1 - \epsilon}{\epsilon}x_0^n.
	\end{equation*}
	
	To have $k \geq 4$, by Fact~\ref{fact: hankel}, we must have 
	$\Hmat_{2}(1,s_1,\dots,s_4) \succeq 0$, which happens if and only if all of 
	its leading principal minors are non-negative, i.e.,
	\begin{equation*}
		1 \geq 0, \qquad \det \begin{pmatrix}
			1 & s_1 \\
			s_1 & s_2 
		\end{pmatrix} \geq 0, \quad \text{and} \quad \det \begin{pmatrix}
			1 & s_1 & s_2\\
			s_1 & s_2 & s_3\\
			s_2 & s_3 & s_4
		\end{pmatrix} \geq 0.
	\end{equation*}
	The first is obviously true. 
	The determinant of $\begin{pmatrix}
		1 & s_1 \\
		s_1 & s_2 
	\end{pmatrix}$ is equal to $s_2 - s_1^2 = \frac{m_2}{\epsilon} - 
	\frac{1-\epsilon}{\epsilon^2}x_0^2$, which is non-negative if and only if 
	$x_0^2 \leq \frac{\epsilon}{1 - \epsilon}m_2 $.  
	
	Observe that the last inequality is exactly $\det 
	\Hmat_{2}(1,s_1,\dots,s_4) \geq 0$, which we now show to be false for any 
	choice of $x_0^2 \leq \frac{\epsilon}{1 - \epsilon}m_2 $, for any $\epsilon 
	< \eta(W)$, which implies, by Fact~\ref{fact: hankel}, that $k \leq 3$. 
	Hence, no $X$ with $H(X) \leq h < \binent(\eta(W))$ has $m_1,m_2,m_3,m_4$ 
	has the first four moments.
	The determinant of $\Hmat_{2}(1,s_1,\dots,s_4)$ is given by 
	\begin{align*}
		\det \begin{pmatrix}
			1 & s_1 & s_2\\
			s_1 & s_2 & s_3\\
			s_2 & s_3 & s_4
		\end{pmatrix}
		& = \det \begin{pmatrix}
			1 &  - \frac{1-\epsilon}{\epsilon}x_0 & \frac{m_2}{\epsilon} -  
			\frac{1 - \epsilon}{\epsilon}x_0^2 \\
			- \frac{1-\epsilon}{\epsilon}x_0  & \frac{m_2}{\epsilon} -  \frac{1 
			- \epsilon}{\epsilon}x_0^2 & \frac{m_3}{\epsilon} -  \frac{1 - 
			\epsilon}{\epsilon}x_0^3 \\
			\frac{m_2}{\epsilon} -  \frac{1 - \epsilon}{\epsilon}x_0^2 &  
			\frac{m_3}{\epsilon} -  \frac{1 - \epsilon}{\epsilon}x_0^3 & 
			\frac{m_4}{\epsilon} -  \frac{1 - \epsilon}{\epsilon}x_0^4
		\end{pmatrix}\\
		&= \frac{1}{\epsilon^3} \det \begin{pmatrix}
			\epsilon & (\epsilon-1)x_0 & m_2 -  (1 - \epsilon)x_0^2 \\
			(\epsilon-1)x_0  & m_2 -  (1 - \epsilon)x_0^2 & m_3 - 
			(1-\epsilon)x_0^3 \\
			m_2 -  (1 - \epsilon)x_0^2 &  m_3 - (1-\epsilon)x_0^3  & m_4 -  (1 
			- \epsilon)x_0^4
		\end{pmatrix}.
	\end{align*}
	Using row reductions to simplify calculations, we get that the above 
	determinant is equal to $\frac{\alpha}{\epsilon^2} - 
	\frac{\beta}{\epsilon^3}$, where 
	\begin{align*}
		\alpha &=  m_2 x_0^4 - 2m_3 x_0^3 + ( m_4 - 3m_2^2 ) x_0^2  + 2 m_2 m_3 
		x_0 + (m_2 m_4 - m_3^2),\\
		\beta &= m_2 x_0^4 - 2m_3 x_0^3 + ( m_4 - 3m_2^2 ) x_0^2  + 2 m_2 m_3 
		x_0 + m_2^3.
	\end{align*}
	Let the polynomial $p(x) = m_2 x^4 - 2m_3 x^3 + ( m_4 - 3m_2^2 ) x^2  + 2 
	m_2 m_3 x$, then we have $\alpha = p(x_0) + (m_2 m_4 - m_3^2)$ and $\beta = 
	p(x_0) + m_2^3$. Note that $x_0^2 \leq \frac{\epsilon}{1 - \epsilon}m_2 $ 
	and $p(0) = 0$. Hence, for sufficiently small $\epsilon$ (say $\epsilon < 
	\epsilon_1$), we have that $\beta = m_2^3 + p(x_0) > \frac{m_2^3}{2}$ for 
	any choice of $x_0$ such that $x_0^2 \leq \frac{\epsilon}{1 - \epsilon}m_2 
	$. Further, $\det \Hmat_{2}(1,s_1,\dots,s_4) = \frac{1}{\epsilon^3}(\alpha 
	\epsilon - \beta)$. For $\epsilon < \epsilon_1$, we have that this 
	expression is at most $\frac{1}{\epsilon^3}(\alpha \epsilon - 
	\frac{m_2^3}{2})$, which is guaranteed to be negative for sufficiently 
	small $\epsilon$ (say $\epsilon < \epsilon_2$). 
	
	Combining the above and taking $\eta(W) = \min\{\epsilon_1, \epsilon_2\}$, 
	for any $\epsilon < \eta(W)$, and $x_0 \in \bbR$, we have $\Hmat_2(1, s_1, 
	\dots, s_4) \not\succeq 0$, since either $\det \begin{pmatrix}
		1 & s_1 \\
		s_1 & s_2 
	\end{pmatrix} < 0$ or $\det \begin{pmatrix}
		1 & s_1 & s_2\\
		s_1 & s_2 & s_3\\
		s_2 & s_3 & s_4
	\end{pmatrix} < 0$. 
	Thus, there is no choice of $\tilde s_5$ and $\tilde s_6$ that makes 
	$\Hmat_3(1, s_1, \dots, s_4, \tilde s_5, \tilde s_6)$ positive 
	semidefinite. By Fact~\ref{fact: hankel}, no $\tilde X$ has $s_1, \dots, 
	s_4$ as the first four moments, and hence, no $X$ with $H(X) < 
	\binent(\eta(W))$ has $m_1, \dots, m_4$ as the first four moments.
	
	We can obtain an explicit expression for $\eta(W)$  in the special case 
	when $W$ is symmetric. In this case, we have $p(x) = m_2 x^4 + (m_4 - 
	3m_2^2) x^2$, $\alpha = p(x_0) + m_2 m_4$, and $\beta = p(x_0) + m_2^3$. 
	The determinant of $\Hmat_{2}(1,s_1,\dots,s_4)$ is given by
	\begin{equation*}
		\det \Hmat_{2}(1,s_1,\dots,s_4) = \frac{1}{\epsilon^3} (\epsilon \alpha 
		- \beta) = \frac{1}{\epsilon^3} \big[ -(1-\epsilon) p(x_0) + (\epsilon 
		m_2 m_4 - m_2^3)\big].
	\end{equation*}
	
	{Case 1: } If $m_4 - 3 m_2^2 \geq 0$, the minimum of $p(x_0)$ is attained 
	at $x_0 = 0$, and hence, the determinant $\det \Hmat_{2}(1,s_1,\dots,s_4) 
	\leq \frac{1}{\epsilon^3} (\epsilon m_2 m_4 - m_2^3)$, which is negative 
	for $\epsilon < \frac{m_2^2}{m_4}$. 
	
	{Case 2: } On the other hand, if $m_4 - 3 m_2^2 < 0$, the minimum of 
	$p(x_0)$ over $x_0^2 \leq \frac{\epsilon}{1 - \epsilon}m_2$ is attained at 
	$x_0^2$ equal to the minimum of $\frac{\epsilon}{1 - \epsilon}m_2$ and 
	$\frac{3m_2^2 - m_4}{2m_2}$ (the latter is the position of the global 
	minimum of $p$). If $\frac{\epsilon}{1 - \epsilon}m_2 \leq \frac{3m_2^2 - 
	m_4}{2m_2}$, or equivalently, $\epsilon \leq \frac{3m_2^2 - m_4}{5m_2^2 - 
	m_4}$, we have $\min_{x_0^2 \leq \frac{\epsilon}{1 - \epsilon}m_2} p\left( 
	x_0\right) = p\left(\sqrt{\frac{\epsilon}{1 - \epsilon}m_2}\right) =  
	\frac{\epsilon^2}{(1 - \epsilon)^2}m_2^3 + (m_4 - 3m_2^2) \frac{\epsilon}{1 
	- \epsilon}m_2$. Hence, the determinant is upper bounded as 
	\begin{align*}
		\det \Hmat_{2}(1,s_1,\dots,s_4) &\leq \frac{1}{\epsilon^3} \left[ 
		-\frac{\epsilon^2}{(1 - \epsilon)}m_2^3 - \epsilon (m_4 - 3m_2^2)  m_2 
		+  (\epsilon m_2 m_4 - m_2^3)\right]\\
		&= \frac{1}{\epsilon^3} \left[ -\frac{\epsilon^2}{(1 - \epsilon)}m_2^3 
		+  3 \epsilon m_2^3  - m_2^3 \right]\\
		& = \frac{m_2^3}{\epsilon^3 (1-\epsilon)} \big[ -\epsilon^2 + 
		(1-\epsilon)(3\epsilon - 1) \big]\\
		& = \frac{m_2^3}{\epsilon^3 (1-\epsilon)} \big[ -\epsilon^2 + 3\epsilon 
		- 1 - 3\epsilon^2 + \epsilon \big]\\
		& = -\frac{m_2^3}{\epsilon^3 (1-\epsilon)} ( 1-2\epsilon )^2,        
	\end{align*}
	which is always negative, as $\epsilon \leq \frac{3m_2^2 - m_4}{5m_2^2 - 
	m_4} < \frac12$. Hence, the above expression is always negative, and we 
	have that $\det \Hmat_{2}(1,s_1,\dots,s_4)$ is always negative for 
	$\epsilon < \frac{3m_2^2 - m_4}{5m_2^2 - m_4}$. Instead, if $\epsilon > 
	\frac{3m_2^2 - m_4}{5m_2^2 - m_4}$, then the minimum of $p(x_0)$ is 
	attained at $x_0^2 = \frac{3m_2^2 - m_4}{2m_2}$, with 
	$p\left(\sqrt{\frac{3m_2^2 - m_4}{2m_2}}\right) = -m_2\left(\frac{3m_2^2 - 
	m_4}{2m_2}\right)^2$. Then, the determinant is upper bounded as 
	\begin{align*}
		\det \Hmat_{2}(1,s_1,\dots,s_4) &\leq \frac{1}{\epsilon^3} \left[ 
		(1-\epsilon)m_2\left(\frac{3m_2^2 - m_4}{2m_2}\right)^2 + (\epsilon m_2 
		m_4 - m_2^3)\right]\\
		& = \frac{1}{4\epsilon^3 m_2^2} \left[ \epsilon(- 9 m_2^4 + 10m_2^2 m_4 
		- m_4^2) - (- 5 m_2^4 + 6m_2^2 m_4 - m_4^2) \right]\\
		& = \frac{m_4 - m_2^2}{4\epsilon^3 m_2^2} \left[ \epsilon(9m_2^2 - m_4) 
		- (5m_2^2 - m_4) \right], 
	\end{align*}
	which is negative for $\epsilon < \frac{5m_2^2 - m_4}{9m_2^2 - m_4}$, as 
	$m_4 > m_2^2$ (these are the first and second moments of the random 
	variable $W^2$). Note that $\frac{5m_2^2 - m_4}{9m_2^2 - m_4} > 
	\frac{3m_2^2 - m_4}{5m_2^2 - m_4}$, hence we have that the determinant is 
	negative if $\frac{3m_2^2 - m_4}{5m_2^2 - m_4} < \epsilon < \frac{5m_2^2 - 
	m_4}{9m_2^2 - m_4}$. Thus, for the case when $m_4 < 3m_2^2$, we have that 
	the determinant is negative if either $\epsilon < \frac{3m_2^2 - 
	m_4}{5m_2^2 - m_4}$ or $\frac{3m_2^2 - m_4}{5m_2^2 - m_4} < \epsilon < 
	\frac{5m_2^2 - m_4}{9m_2^2 - m_4}$, which is equivalent to $\epsilon < 
	\frac{5m_2^2 - m_4}{9m_2^2 - m_4}$.
	
	A summary of the above is as follows: If $m_4 \geq 3m_2^2$, then for every 
	$\epsilon < \frac{m_2^2}{m_4}$, we have either $\det 
	\Hmat_2(1,s_1,\dots,s_4) \leq 0$ (by choosing $x_0$ such that $x_0^2 \leq 
	\frac{\epsilon}{1 - \epsilon}m_2$) or $\det \Hmat_1(1,s_1,s_2) \leq 0$ (for 
	$x_0^2 > \frac{\epsilon}{1 - \epsilon}m_2$). Similarly, if $m_4 < 3m_2^2$, 
	for every $\epsilon < \frac{5m_2^2 - m_4}{9m_2^2 - m_4}$, we again have 
	either $\det \Hmat_2(1,s_1,\dots,s_4) \leq 0$ for $x_0^2 \leq 
	\frac{\epsilon}{1 - \epsilon}m_2$) or $\det \Hmat_1(1,s_1,s_2) \leq 0$ (for 
	$x_0^2 > \frac{\epsilon}{1 - \epsilon}m_2$). Thus, by defining $\eta(W)$ as 
	\begin{equation*}
		\eta(W) = \begin{cases}
			\frac{m_2^2}{m_4} & \text{if } m_4 \geq 3 m_2^2,\\
			\frac{5m_2^2 - m_4}{9m_2^2 - m_4} & \text{if } m_4 <3 m_2^2,
		\end{cases}
	\end{equation*}
	for any $\epsilon < \eta(W)$, we have $\Hmat_{2}(1,s_1,\dots,s_4) 
	\not\succeq 0$ for any $x_0 \in \bbR$. 
	Hence, no $X$ with $H(X) < \binent(\eta(W))$ has $m_1, m_2,\dots, m_4$ as 
	the first four moments. Note that $\eta < \frac12$ always, as for $m_4 \geq 
	3m_2^2$, we have that $\eta(W) = \frac{m_2^2}{m_4} \leq \frac13$, and for 
	$m_4 < 3 m_2^2$, we have $\frac{5m_2^2 - m_4}{9m_2^2 - m_4} < \frac12$, 
	since $m_4 > m_2^2$. 
	
	\paragraph{Proof of (ii): } Let $h > 0$ be arbitrary. We are to show that 
	there exists $X$ with $H(X) \leq h$ with $\bbE[X^n] = m_n$ for $n = 1,2,3$. 
	By Lemma~\ref{lem: low-ent}, such an $X$ exists if and only if there is 
	some $x_0 \in \bbR$, $\epsilon \in \left(0, \frac12\right)$ such that 
	$\binent(\epsilon) \leq h$, and $\tilde X$ with $H(\tilde X) \leq \frac{h - 
	\binent(\epsilon)}{\epsilon}$. Let $\epsilon$ be such that $2 
	\binent(\epsilon) = h$, then we must have $H(\tilde X) \leq 
	\frac{\binent(\epsilon)}{\epsilon}$, which is more than $2\log 2$ for all 
	$\epsilon \in \left(0,\frac12\right)$.
	We show that for every $\epsilon \in (0, \eta(W)) \subseteq 
	\left(0,\frac12\right)$, there exists some choice of $x_0 \in \bbR$ and 
	$\tilde s_4 \in \bbR$ such that the Hankel matrix $\Hmat_2(1, s_1, s_2, 
	s_3, \tilde s_4)$ is positive semidefinite. Fact~\ref{fact: hankel} then 
	guarantees the existence of some $\tilde X$ with at most two atoms and 
	$s_1, s_2, s_3$ as the first three moments. Since $\tilde X$ is supported 
	on at most two atoms, $H(\tilde X) \leq \log 2  < 
	\frac{\binent(\epsilon)}{\epsilon}$.  Hence, by Lemma~\ref{lem: low-ent}, 
	there exists a discrete random variable $X$ with at most three atoms, $H(X) 
	\leq h$, and $m_1, m_2, m_3$ as the first three moments. 
	
	Again, the matrix $\Hmat_2(1, s_1, s_2, s_3, \tilde s_4)$ is positive 
	semidefinite if and only if all of its leading principal minors are 
	non-negative, i.e.,  
	\begin{equation*}
		1 \geq 0, \qquad \det \begin{pmatrix}
			1 & s_1 \\
			s_1 & s_2 
		\end{pmatrix} \geq 0, \quad \text{and} \quad \det \begin{pmatrix}
			1 & s_1 & s_2\\
			s_1 & s_2 & s_3\\
			s_2 & s_3 & \tilde s_4
		\end{pmatrix} \geq 0.
	\end{equation*}
	The first of these is trivially true. As seen in the proof of the previous 
	part, we can ensure that $\det \begin{pmatrix}
		1 & s_1 \\
		s_1 & s_2 
	\end{pmatrix} > 0$ by choosing $x_0$ such that $x_0^2 < 
	\frac{\epsilon}{1-\epsilon} m_2$. The third can also be ensured by choosing 
	a sufficiently large $\tilde s_4$, as the determinant is a linear function 
	of $\tilde s_4$ with a positive coefficient (this can be seen by expanding 
	the determinant along the third row or column). Since this can be done for 
	any arbitrarily small $\epsilon > 0$, we are done.        \hfill 
	$\blacksquare$

	\section{Discussion and Conclusion} \label{sec: conc}
	
	We considered the problem of computing the capacity for a power-constrained 
	Gaussian channel with an input entropy constraint. We characterized this 
	capacity in asymptotic regimes of low and high entropy at all (constant) 
	values of SNR, and low SNR at sufficiently small entropy. However, 
	identifying a capacity-achieving distribution, which is guaranteed to exist 
	by Proposition~\ref{prop: exist}, remains an open problem in all regimes. 
	Even obtaining non-trivial upper and lower bounds for $C_H$ at intermediate 
	values of entropy and SNR would be useful. The major difficulty in solving 
	the optimization problem defining the capacity is that the entropy 
	constraint makes the problem non-convex. In fact, the feasible set is the 
	complement of a convex set, known as a \emph{reverse convex} constraint in 
	the optimization literature \cite{tuy1987convex}.
	
	An estimation-theoretic interpretation of $C_H$ is the following. By the 
	I-MMSE relationship \cite{guo2004immse}, we have $I(X, \snr) = \frac12 
	\int_0^\snr \mmse(X, \gamma) \d \gamma$ and $H(X) = \frac12 \int_0^\infty 
	\mmse(X, \gamma) \d \gamma$, where $\mmse(X, \gamma)$ is the minimum mean 
	squared error (MMSE) of estimating $X$ from $Y = \sqrt{\gamma} X + Z$. 
	Hence, in computing $C_H$, we consider all distributions on $X$ that have 
	the same total integral under the curve $\gamma \mapsto \mmse(\gamma)$ over 
	$[0, \infty)$, and choose one which maximizes the value of the integral 
	over the range $[0,\snr]$. This implies that the maximizing $X$ should have 
	a large MMSE at small SNR and vice-versa, with the transition as sharp as 
	possible at SNR equal to $\snr$.
	
	To characterize $C_H$ in the low-SNR regime at small entropy, we first 
	solved a low-entropy version of the moment problem.
	We know that it is necessary to have the cardinality of the support grow to 
	infinity to match arbitrarily many moments, but
	Theorem~\ref{thm: main} shows that this is not sufficient. 
	It is necessary to have a discrete distribution of non-vanishing entropy to 
	match arbitrarily many moments of a continuous distribution, even if the 
	cardinality of the support is allowed to be arbitrarily large. 
	Interestingly, if the entropy is a sufficiently small (but even 
	non-vanishing) constant, no more than three moments can be matched.
	Recall the Gaussian example: the Gauss--Hermite quadrature provides an 
	$m$-point discrete distribution that has an entropy which grows 
	approximately as $\frac{1}{2}\log m$ and matches $2m-1$ moments. This 
	solution is optimal in the sense that no other distribution of $m$ points 
	or fewer can match $2m-1$ moments, but there are infinitely many solutions 
	with more than $m$ points.
	Another interesting question is then the following: is it necessary for the 
	entropy to grow to infinity to match arbitrarily many moments with a 
	continuous distribution? Equivalently, 
	does there exist $h > 0$ such that there are $m$-point distributions which 
	match $k_m$ moments (for some sequence $(k_m)_{m = 0}^\infty$ which goes to 
	infinity as $m \to \infty$, 
	possibly with $k_m < 2m-1$) but such that the entropy is uniformly bounded 
	above by $h$ for all $m$?

	The answer turns out to be no, for the following reason. Suppose we had a 
	sequence of discrete random variables $(X_m)_{m=0}^{\infty}$ such that 
	$X_m$ is supported on $m$ points and has the same first $k_m$ moments as 
	the Gaussian distribution, with $k_m \to \infty$ as $m \to \infty$. Then we 
	must have that $X_m$ converges in distribution to $X$ that is Gaussian 
	\cite[Theorem 30.2]{billingsley1995measure}. The lower-semicontinuity of 
	entropy then implies that $\liminf_m H(X_m) \geq H(X)$, which is infinity, 
	and we are done.
	This confirms the heuristic statement that it is entropy and not 
	cardinality that helps match moments, and hence, replicate tail behaviour 
	of continuous distributions.

	\section*{Acknowledgment}
	The work of S.~Shamai was supported by the German Research Foundation (DFG) 
	via
	the German--Israeli Project Cooperation (DIP), under Project SH 1937/1--1.

	% \bibliographystyle{IEEEtran}
	% \bibliography{ref}
	
	\printbibliography

@book{Schmudgen2017book,
	title        = {The moment problem},
	author       = {Konrad Schmüdgen},
	series       = {Graduate Texts in Mathematics},
	publisher    = {Springer},
	year         = {2017},
	doi          = {10.1007/978-3-319-64546-9},
	isbn         = {978-3-319-64545-2},
	eisbn        = {978-3-319-64546-9},
	pages        = {XXIII, 512},
	volume          = {277}
}

@book{akhiezer1965classical,
	title={The classical moment problem and some related questions in analysis},
	author={Akhiezer, Naum Ilji{\v{c}}},
	year={1965},
	publisher={Oliver \& Boyd Edinburgh}
}

@article{CurtoFialkow1991,
	author    = {Raul E. Curto and Lawrence A. Fialkow},
	title     = {Recursiveness, positivity, and truncated moment problems},
	journal   = {Houston Journal of Mathematics},
	volume    = {17},
	number    = {4},
	year      = {1991},
	pages     = {603--635}
}

@article{hamburger1920erweiterung,
	title={{\"U}ber eine {E}rweiterung des {S}tieltjesschen {M}omentenproblems},
	author={Hamburger, Hans},
	journal={Mathematische Annalen},
	volume={81},
	number={2},
	pages={235--319},
	year={1920},
	publisher={Springer}
}

@article{hausdorff1921summationsmethoden,
	title={Summationsmethoden Und {M}omentfolgen. {I}},
	author={Hausdorff, Felix},
	journal={Mathematische Zeitschrift},
	volume={9},
	number={1},
	pages={74--109},
	year={1921},
	publisher={Springer}
}

@inproceedings{stieltjes1894recherches,
	title={Recherches sur les fractions continues},
	author={Stieltjes, T-J},
	booktitle={Annales de la Facult{\'e} des sciences de Toulouse: 
	Math{\'e}matiques},
	volume={8},
	number={4},
	pages={J1--J122},
	year={1894}
}

@book{stoer1980introduction,
	title={Introduction to numerical analysis},
	author={Stoer, Josef and Bulirsch, Roland},
	year={2002},
	publisher={Springer}
}

@inproceedings{wu2010impact,
	title={The impact of constellation cardinality on {G}aussian channel 
	capacity},
	author={Wu, Yihong and Verd{\'u}, Sergio},
	booktitle={48th Annual Allerton Conference on Communication, Control, and 
	Computing (Allerton)},
	pages={620--628},
	year={2010},
	organization={IEEE}
}

@book{gauss1814methodus,
	title={Methodus nova integralium valores per approximationem inveniendi},
	author={Gauss, Carl Friedrich},
	year={1814}
}

@article{smith1971information,
	title={The information capacity of amplitude- and variance-constrained 
	scalar {G}aussian channels},
	author={Smith, Joel G.},
	journal={Information and control},
	volume={18},
	number={3},
	pages={203--219},
	year={1971},
	publisher={Elsevier}
}

@ARTICLE{afaycal2001fading,
	title={The capacity of discrete-time memoryless {R}ayleigh-fading 
	channels}, 
	author={Abou-Faycal, Ibrahim C. and Trott, Mitchell D. and Shamai, Shlomo},
	journal={IEEE Transactions on Information Theory}, 
	year={2001},
	volume={47},
	number={4},
	pages={1290-1301},
	doi={10.1109/18.923716}}

@book{luenberger1997optimization,
	title={Optimization by vector space methods},
	author={Luenberger, David G.},
	year={1997},
	publisher={John Wiley \& Sons}
}

@book{billingsley2013convergence,
	title={Convergence of probability measures},
	author={Billingsley, Patrick},
	year={2013},
	publisher={John Wiley \& Sons}
}

@book{billingsley1995measure,
	title={Probability and measure},
	author={Billingsley, Patrick},
	year={1995},
	publisher={John Wiley \& Sons}
}

@ARTICLE{calmon2018sdpi,
	title={Strong data processing inequalities for input constrained additive 
	noise channels}, 
	author={Calmon, Flavio P. and Polyanskiy, Yury and Wu, Yihong},
	journal={IEEE Transactions on Information Theory}, 
	year={2018},
	volume={64},
	number={3},
	pages={1879-1892},
	doi={10.1109/TIT.2017.2782359}}

@article{mccullagh1994does,
	title={Does the moment-generating function characterize a distribution?},
	author={McCullagh, Peter},
	journal={The American Statistician},
	volume={48},
	number={3},
	pages={208--208},
	year={1994},
	publisher={Taylor \& Francis}
}

@INPROCEEDINGS{guo2004immse, 
	title={Mutual information and {MMSE} in {G}aussian channels}, 
	author={Dongning Guo and Shamai, Shlomo and Verd\'u, Sergio},
	booktitle={Proc.\ IEEE International Symposium on Information Theory 
	(ISIT)},
	year={2004},
	volume={},
	number={},
	pages={349-349},
	doi={10.1109/ISIT.2004.1365386}}

@INPROCEEDINGS{guo2008estimation,
	title={Estimation of non-{G}aussian random variables in {G}aussian noise: 
	Properties of the {MMSE}}, 
	author={Guo, Dongning and Shamai, Shlomo and Verd\'u, Sergio},
	booktitle={Proc.\ IEEE International Symposium on Information Theory 
	(ISIT)}, 
	year={2008},
	volume={},
	number={},
	pages={1083-1087},
	doi={10.1109/ISIT.2008.4595154}}

@article{sharma2010transition,
	title={Transition points in the capacity-achieving distribution for the 
	peak-power limited {AWGN} and free-space optical intensity channels},
	author={Sharma, Naresh and Shamai, Shlomo},
	journal={Problems of Information Transmission},
	volume={46},
	number={4},
	pages={283--299},
	year={2010},
	publisher={Springer}
}

@article{dytso2019capacity,
	title={The capacity achieving distribution for the amplitude constrained 
	additive {G}aussian channel: An upper bound on the number of mass points},
	author={Dytso, Alex and Yagli, Semih and Poor, H.\ Vincent and Shlomo 
	Shamai},
	journal={IEEE Transactions on Information Theory},
	volume={66},
	number={4},
	pages={2006--2022},
	year={2019},
	publisher={IEEE}
}

@article{thangaraj2017capacity,
	title={Capacity bounds for discrete-time, amplitude-constrained, additive 
	white {G}aussian noise channels},
	author={Thangaraj, Andrew and Kramer, Gerhard and B{\"o}cherer, Georg},
	journal={IEEE Transactions on Information Theory},
	volume={63},
	number={7},
	pages={4172--4182},
	year={2017},
	publisher={IEEE}
}

@inproceedings{ma2021first,
	title={First- and second-moment constrained {G}aussian channels},
	author={Ma, Shuai and Wigger, Mich{\`e}le},
	booktitle={2021 IEEE International Symposium on Information Theory (ISIT)},
	pages={432--437},
	year={2021}
}

@ARTICLE{lapidoth2009poisson,
	title={On the capacity of the discrete-time {P}oisson channel}, 
	author={Lapidoth, Amos and Moser, Stefan M.},
	journal={IEEE Transactions on Information Theory}, 
	year={2009},
	volume={55},
	number={1},
	pages={303-322},
	doi={10.1109/TIT.2008.2008121}}

@INPROCEEDINGS{barletta2024binomial, 
	title={Binomial channel: On the capacity-achieving distribution and bounds 
	on the capacity},
	author={Barletta, Luca and Zieder, Ian and Favano, Antonino and Dytso, 
	Alex},
	booktitle={Proc.\ IEEE International Symposium on Information Theory 
	(ISIT)}, 
	year={2024},
	volume={},
	number={},
	pages={711-716},
	doi={10.1109/ISIT57864.2024.10619601}}

@ARTICLE{liu2022lossy,
	title={Cross-domain lossy compression as entropy constrained optimal 
	transport}, 
	author={Liu, Huan and Zhang, George and Chen, Jun and Khisti, Ashish},
	journal={IEEE Journal on Selected Areas in Information Theory}, 
	year={2022},
	volume={3},
	number={3},
	pages={513-527},
	doi={10.1109/JSAIT.2022.3229670}}

@inproceedings{
	ebrahimi2024minimum,
	title={Minimum entropy coupling with bottleneck},
	author={Mohammad Reza Ebrahimi and Jun Chen and Ashish Khisti},
	booktitle={The Thirty-eighth Annual Conference on Neural Information 
	Processing Systems (NeurIPS)},
	year={2024}
}

@ARTICLE{alghamdi2024moments,
	title={Measuring information from moments}, 
	author={Alghamdi, Wael and Calmon, Flavio P.},
	journal={IEEE Transactions on Information Theory}, 
	year={2024},
	volume={70},
	number={2},
	pages={763-802},
	doi={10.1109/TIT.2022.3202492}}

@inproceedings{li2015generative,
	title={Generative moment matching networks},
	author={Li, Yujia and Swersky, Kevin and Zemel, Rich},
	booktitle={International conference on machine learning (ICML)},
	pages={1718--1727},
	year={2015},
	organization={PMLR}
}

@inproceedings{nguyen2021distributional,
	title={Distributional reinforcement learning via moment matching},
	author={Nguyen-Tang, Thanh and Gupta, Sunil and Venkatesh, Svetha},
	booktitle={Proceedings of the AAAI Conference on Artificial Intelligence},
	volume={35},
	number={10},
	pages={9144--9152},
	year={2021}
}

@article{tuy1987convex,
	title={Convex programs with an additional reverse convex constraint},
	author={Tuy, Hoang},
	journal={Journal of optimization theory and applications},
	volume={52},
	number={3},
	pages={463--486},
	year={1987},
	publisher={Springer}
}

@ARTICLE{wu2012mmse,
	title={Functional properties of minimum mean-square error and mutual 
	information}, 
	author={Wu, Yihong and Verd\'u, Sergio},
	journal={IEEE Transactions on Information Theory}, 
	year={2012},
	volume={58},
	number={3},
	pages={1289-1301},
	doi={10.1109/TIT.2011.2174959}}
	
\end{document}